\documentclass[a4paper,11pt,DIV=11,oneside,headings=small,bibliography=totoc,abstract=on]{scrartcl}
\usepackage{amssymb,amsmath,amsfonts,amsthm}
\usepackage[markup=default]{changes}
\usepackage[english]{babel}
\usepackage[utf8]{inputenc}
\usepackage{enumerate,tikz}
\newcommand\bigtimes{%
  \ensuremath\mathop{%
    \tikz[baseline,line width=0.8pt,line cap=round]%
      \draw(0,-0.2em)--(0.9em,0.7em)(0.9em,-0.2em)--(0,0.7em);%
  }%
}
\usepackage[noblocks]{authblk}

\newcommand{\sfUC}{\mathrm{sfuc}}
\newcommand{\Dom}{\mathcal{D}}
\newcommand{\diver}{\operatorname{div}}
\renewcommand{\epsilon}{\varepsilon}
\newcommand{\Op}{\mathcal{H}}

\newcommand{\T}{\mathrm{T}}

\newcommand{\RR}{\mathbb{R}}
\newcommand{\CC}{\mathbb{C}}
\newcommand{\NN}{\mathbb{N}}    
\newcommand{\ZZ}{\mathbb{Z}}
\newcommand{\PP}{\mathbb{P}}
\newcommand{\EE}{\mathbb{E}}
\newcommand{\ii}{\mathrm{i}}
\newcommand{\ball}[2]{B (#1 , #2)}
\newcommand{\ballc}[1]{B (#1)}
\newtheorem{theorem}{Theorem}[section]
\newtheorem{lemma}[theorem]{Lemma}

\theoremstyle{definition}

\theoremstyle{remark}
\newtheorem{remark}[theorem]{Remark}
\usepackage{microtype}
\usepackage{ellipsis}
\usepackage[textsize=tiny,shadow]{todonotes}
\begin{document}
%
%
%
%
%
%
%
%
\title{Wegner estimate and disorder dependence for alloy-type Hamiltonians with bounded magnetic potential}
\author[1]{Matthias T\"aufer}
\author[2]{Martin Tautenhahn}
\affil[1]{Technische Universit\"at Dortmund, Fakult\"at f\"ur Mathematik, Germany}
\affil[2]{Friedrich Schiller Universit\"at Jena, Germany \& Technische Universit\"at Chemnitz, Germany}
\date{\vspace{-2em}}
\maketitle
\begin{abstract}
 We consider non-ergodic magnetic random Schr{\"o}dinger operators with a bounded magnetic vector potential. We prove an optimal Wegner estimate valid at all energies. The proof is an adaptation of arguments from \cite{Klein-13}, combined with a recent quantitative unique continuation estimate for eigenfunctions of elliptic operators from \cite{BorisovTV-15}. 
 This generalizes Klein's result to operators with a bounded magnetic vector potential.
 Moreover, we study the dependence of the Wegner-constant on the disorder parameter. In particular, we show that above the model-dependent threshold $E_0(\infty) \in (0, \infty]$, it is impossible that the Wegner-constant tends to zero if the disorder increases. This result is new even for the standard (ergodic) Anderson Hamiltonian without magnetic field. 
\end{abstract}

%
%
%
%
%
%
%
%
\section{Introduction}
We study a class of magnetic, non-ergodic random Schr\"odinger operators on $L^2 (\RR^d)$ of the type
\begin{equation}
\label{eq:model}
 H_\omega = H_0 + \lambda V_\omega ,
\end{equation}
where $\lambda > 0$ is the disorder parameter, $H_0 = ( - \ii \nabla + A_0)^2 + V_0$ with a bounded electric potential $V_0 \in L^\infty(\RR^d)$ and a bounded magnetic vector potential $A_0 \in L^\infty(\RR^d, \RR^d)$ satisfying $\diver (A_0)$ bounded, and where 
 \[
   V_\omega(x) 
   =
   \sum_{j \in \ZZ^d} \omega_j u_j(x - z_j ) .
  \]
The random variables $\omega_j$, $j \in \ZZ^d$, are independent and identically distributed with compactly supported, non-degenerate distribution, and the single-site potentials $( u_j )_{j \in \ZZ^d}$, are measurable and real-valued functions on $\RR^d$ satisfying
   \[
    u_- \chi_{\ballc{\delta_-}} 
    \leq 
    u_j 
    \leq 
    \chi_{\Lambda_{\delta_+}}
    \]
for some $u_- \in (0,1]$ and $\delta_\pm > 0$, where $\ballc{\delta_-}$ and $\Lambda_{\delta_+}$ denote the $d$-dimensional ball of radius $\delta_-$ and the $d$-dimensional cube of side length $\delta_+$, centered at $0$. Furthermore, we assume that the centers $z_j$ of the single-site potentials are in a certain sense equidistributed in $\RR^d$. 
Such operators are used to model quantum mechanical properties of disordered solids. While each configuration of the randomness corresponds to a particular realization of the solid, the law of the random variables models their distribution.
\par
One distinctive feature of random operators is the phenomenon of localization, i.e.\ that parts of the spectrum consist only of pure point spectrum (spectral localization) or that the solutions of the Schr\"odinger equation stay almost surely trapped in a finite region of space for all time (dynamical localization). This is in contrast to periodic operators which exhibit only absolutely continuous spectrum. One method for proving localization is the so-called multiscale analysis introduced in \cite{FroehlichS-83,FroehlichMSS-85} and further developed in \cite{DreifusK-89,GerminetK-01,GerminetK-03,GerminetK-06}. The multiscale analysis is an induction argument. 
While the induction anchor is provided by the so-called initial-scale estimate, a so-called Wegner estimate is needed for the induction step. 
A Wegner estimate is an upper bound on the expected number of eigenvalues of a self-adjoint restriction
 $H_{\omega , L}$ of $H_\omega$ to a cube $\Lambda_L \subset \RR^d$ of side length $L>0$ in an energy interval $[a,b] \subset \RR$. 
 More precisely, a Wegner estimate is an estimate of the form
\begin{equation}
\label{eq:Wegner_introduction}
 \EE \bigl( \operatorname{Tr} \chi_{[a,b]} (H_{\omega , L}) \bigr)
 \leq
 C(\lambda) S (b-a) \lvert \Lambda_L \rvert^m ,
\end{equation}
where $C(\lambda)$ is a constant depending on $\lambda$ and the various model parameters, $S (b-a)$ denotes the concentration function of the distribution of the random variables $\omega_j$, $j \in \ZZ^d$, and $m \geq 1$.
A Wegner estimate is optimal if $m = 1$.
\par
In this note, we prove optimal Wegner estimates for the family of operators defined in~\eqref{eq:model} valid at all energies $[a,b] \subset \RR$, see Theorems~\ref{thm:Klein1} and~\ref{thm:Klein2}.
Furthermore, we study the dependence of $C(\lambda)$ on the disorder parameter $\lambda$.
If
\begin{equation*}
b < E_0 (\infty)  := \lim_{t \to \infty} \inf \sigma \biggl( H_0 + t \sum_{j \in \ZZ^d} u_j (\cdot - j) \biggr)
\in (0, \infty]
\end{equation*}
then the constant $C(\lambda)$ in the Wegner estimate~\eqref{eq:Wegner_introduction} will tend to zero if the disorder $\lambda$ tends to infinity, see Theorem~\ref{thm:Klein3}. 
%
%
In this case, the Wegner estimate can be used to obtain the initial length scale estimate at sufficiently large disorder, and localization follows via multiscale analysis, see \cite{DreifusK-89,Kirsch-08}. 
Hence, it is natural to ask whether such a Wegner estimate (where $C(\lambda) \to 0$ if $\lambda \to \infty$) also holds above the threshold $E_0 (\infty)$.
Until now there have been no results in this direction, cf.\ \cite{Klein-13}. 
We show in Theorem~\ref{thm:counterexample} that it is indeed impossible to obtain a Wegner estimate with $C(\lambda) \to 0$ if $\lambda \to \infty$ above $E_0 (\infty)$. 
This result is even new in the absence of a magnetic field.
In the case where the family $H_\omega$ is ergodic, Theorem~\ref{thm:counterexample} has an interpretation in terms of the integrated density of states, see Theorem~\ref{thm:conterexample_ergodic}.
Theorems~\ref{thm:counterexample} and~\ref{thm:conterexample_ergodic} show in particular that the spectral behaviour at large disorder changes drastically below and above the model-dependent constant $E_0(\infty) \in [0, \infty]$, see Remark~\ref{rem:phase_transition}.
\par
Optimal and non-optimal Wegner estimates for ergodic and non-ergodic random operators with and without magnetic field have been studied by many authors before. Let us give a brief overview here. 
We refer to~\cite{Klein-13} and \cite{Veselic-08} for further references.
In~\cite{CombesH-94}, an optimal Wegner estimate at all energies is proved for the usual (ergodic) Anderson Hamiltonian under the additional assumption that a covering condition holds.
In~\cite{CombesHK-07}, the authors remove the covering condition and consider operators $H_0$ with non-vanishing magnetic field.
More precisely, they assume that the magnetic vector potential $A_0$ is either periodic (which allows only for a subset of the class of periodic magnetic fields) or $H_0$ is the Landau Hamiltonian in two dimensions, i.e. $A_0 = B/2 (-x_2, x_1)$ where $B$ is the magnetic field strength.
The Landau Hamiltonian can be treated in spite of its unbounded vector potential $A_0$ since unique continuation properties for the Landau Hamiltonian are quite well understood, cf.~\cite{CombesHKR-04}. 
However, this method strongly relies on the structure of the Landau Hamiltonian and breaks down for arbitrarily small perturbations of the Landau magnetic field.
In the non-ergodic setting, optimal Wegner estimates have been proved, e.g., in \cite{RojasMolina-12} for the Landau Hamiltonian, and up to a logarithmic correction in \cite{RojasMolinaV-13} for non-magnetic Schr\"odinger operators.
The logarithmic factor in the energy is removed in~\cite{Klein-13}.
\par 
 The drawback of our results is that we assume boundedness of $A_0$ and $\diver (A_0)$. 
 This assumption stems from the quantitative unique continuation result from \cite{BorisovTV-15}.
 While the long-term goal is to treat also unbounded magnetic vector potentials $A_0$, we emphasize that already bounded $A_0$ include important physically relevant examples, cf.\ \cite{Krakovsky-96}.
\par
Our paper is organized as follows: In Section~\ref{sect:model_and_results} we introduce the notation and state our main results. In Section~3 we quote a quantitative unique continuation estimate for solutions of elliptic equations given in \cite{BorisovTV-15}. In Section~\ref{sec:proofs} we give the proofs of our main results, while a technical estimate is postponed to the appendix.
The proof of the Wegner estimates from Theorems~\ref{thm:Klein1}, \ref{thm:Klein2} and \ref{thm:Klein3} rely on the recent quantitative unique continuation result of \cite{BorisovTV-15}. 
By now there is a wealth of results pursuing the connection between unique continuation and spectral theory of random Schr\"odinger operators, see e.g.\ \cite{CombesHK-03,CombesHKR-04,BourgainK-05,CombesHK-07,Stollmann-10,BoutetDeMonvelLS-11,RojasMolina-12,RojasMolinaV-13,BourgainK-13,Klein-13,TaeuferV-15,NakicTTV-15,NakicTTV-16-arxiv,TaeuferV-16}.
While the proofs of Theorems~\ref{thm:Klein1}, \ref{thm:Klein2} and \ref{thm:Klein3} use in particular the guiding thread of \cite{Klein-13}, 
they show how this approach and the recent result of \cite{BorisovTV-15} complement each other in an efficient way. 
On the contrary, the proofs of Theorem~\ref{thm:counterexample} and \ref{thm:conterexample_ergodic} are original even for the usual (ergodic) Anderson Hamiltonian with vanishing magnetic field.
%
%
%
%
%
%
%
%
\section{Model and results}
\label{sect:model_and_results}
For $x \in \RR^d$ or $\CC^d$
we denote by $\lvert x \rvert$ the Euclidean norm of $x$. 
For $L,r > 0$ we denote by $\Lambda_L = (- L/2, L/2)^d$ the open cube with side length $L$, and by $\ballc{r} = \{y \in \RR^d \colon \lvert y \rvert < r\}$ the open ball in $\RR^d$ of radius $r$, centered at $0$. 
For $x \in \RR^d$ we denote by $\Lambda_L (x) = \Lambda_L + x$ and $\ball{r}{x} = \ballc{r} + x$ its translates.
\par
Let us define the class of random Schr\"odinger operators studied in this note. It is a generalization of the models studied in \cite{Klein-13}. 
The non-random part is given by the self-adjoint magnetic Schr\"odinger operator 
\[
H_0 = ( - \ii \nabla + A_0)^2 + V_0
\]
on $L^2 (\RR^d)$ with a bounded electric potential $V_0 \in L^\infty(\RR^d)$ and a bounded magnetic vector potential $A_0 \in L^\infty(\RR^d, \RR^d)$ such that $\diver (A_0)$ is bounded and $\inf \sigma(H_0) = 0$. Note that we can rewrite $H_0$ to $H_0 = - \Delta + b_0^T \nabla + c_0$ where
 \begin{equation}
 \label{eq:definition_b_c}
  b_0(x) = - 2 i A_0(x)
  \quad
  \text{and}
  \quad
  c_0(x) = V_0(x) + \lvert A_0(x) \rvert^2 - i \diver(A_0)(x).
 \end{equation}
In order to introduce the random part, we define the probability space $(\Omega , \mathcal{A}, \allowbreak \PP)$ where
\[
 \Omega = \bigtimes\limits_{k \in \ZZ^d} \RR, \quad \mathcal{A} = \bigotimes_{k \in \ZZ^d} \mathcal{B} (\RR), \quad \text{and} \quad \PP = \bigotimes_{k \in \ZZ^d} \mu_k , 
\]
where $\mu_k$, $k \in \ZZ^d$, are non-degenerate probability measures on $\RR$ with $\operatorname{supp} \mu_k \allowbreak \subset [0,M]$ for some $M > 0$ and all $k \in \ZZ^d$. By non-degenerate we mean that for all $L > 0$ we have $S_L (t) \to 0$ as $t \to 0$, see below for the definition of $S_L$. As a consequence, the projections $\Omega \ni \omega \mapsto \omega_k$, $k \in \ZZ^d$ give rise to the independent sequence of random variables $(\omega_k)_{k \in \ZZ^d}$, each coordinate $\omega_k$ distributed according to the measure $\mu_k$. We write $S_\mu(t) := \sup_{a \in \RR} \mu([a, a+t])$ for the concentration function of a probability measure $\mu$ and set for $t \geq 0$
\[
S_L(t) := \sup_{j \in \Lambda_L \cap \ZZ^d} S_{\mu_j}(t) .
\]
We use the symbol $\EE$ for the expectation with respect to the probability measure $\PP$.
\par
Let now $\delta_- \in (0,1/2)$ and $Z = (z_j)_{j \in \ZZ^d} \subset \RR^d$ such that
 \[
  \forall j \in \ZZ^d \colon \quad  \ball{\delta_-}{z_j} \subset \Lambda_1(j).
 \]
For each $\omega \in \Omega$, the \emph{crooked alloy-type potential} $V_\omega : \RR^d \to \RR$ is defined by
  \[
   V_\omega(x) 
   =
   \sum_{j \in \ZZ^d} \omega_j u_j(x - z_j ) ,
  \]
where the single-site potentials $( u_j )_{j \in \ZZ^d}$, are measurable and real-valued functions on $\RR^d$ satisfying
   \[
    u_- \chi_{\ballc{\delta_-}} 
    \leq 
    u_j 
    \leq 
    \chi_{\Lambda_{\delta_+}(0)}
    \]
for some $u_- \in (0,1]$ and $\delta_+ > 0$. 
For each $\omega \in \Omega$ and $\lambda > 0$ we define the self-adjoint operator 
 \[
  H_\omega = H_0 + \lambda V_\omega 
 \]
in $L^2 (\RR^d)$, and call the family of operators $( H_\omega)_{\omega \in \Omega}$ the \emph{magnetic crooked alloy-type Hamiltonian}.
For $L > 0$ we denote by $H_{\omega,L}$ the restrictions of $H_\omega$ to $L^2 (\Lambda_L)$ subject to Dirichlet boundary conditions.
Given a Borel set $B \subset \RR$, we denote by $P_{\omega, L}(B) := \chi_B(H_{\omega, L})$ the spectral projection onto the set $B$ with respect to $H_{\omega,L}$.
\begin{theorem}
\label{thm:Klein1}
 Let $E_0 > 0$ and set
 \begin{equation}
 \label{eq:gamma_1}
  \gamma_1^2 
  = 
  \frac{1}{2} \delta_-^{N_2 \left( 1 + \lvert E_0 \rvert^{2/3} + \lVert b_0 \rVert_\infty^2 + \lVert c_0 \rVert_\infty^{2/3} \right)}
 \end{equation}
 where $M_d > 0$ is the constant from Theorem~\ref{thm:qucp2}
 and $b_0, c_0$ are as in Eq.~\eqref{eq:definition_b_c}.
 Then there is a constant 
 $
  C_1 = C_1 (d, \delta_\pm, u_-, \gamma_1, \lVert V_0 \rVert_\infty, E_0)
 $
 such that for any closed interval $I \subset (- \infty, E_0]$ with $\lvert I \rvert \leq 2 \gamma_1$,
 and $\lambda > 0$, 
 and any $L \in \NN_{\mathrm{odd}}$ with $L \geq 2 + \delta_+$, we have 
 \[
  \EE \bigl( \operatorname{Tr} P_{\omega, L}(I) \bigr)
  \leq
  C_1 \left( 1 + (\lambda M)^{2^{2 + \frac{\log d}{\log 2}}} \right)
  S_L( \lambda^{-1} \lvert I \rvert) \lvert \Lambda_L \rvert.
 \]
 \end{theorem}
 
 \begin{theorem}
 \label{thm:Klein2}
 Let $E_0 >0$ and set
 \begin{equation}
  \label{eq:gamma_2}
  \gamma_2
  = 
  \frac{1}{2} \delta_-^{N_2 \left( 1 + \lvert E_0 \rvert^{2/3} + \lVert b_0 \rVert_\infty^2 + \left( \lVert c_0 \rVert_\infty + \lambda M (2 + \delta_+)^d \right)^{2/3} \right)}
 \end{equation}
 where $N_2 > 0$ is the constant from Theorem~\ref{thm:qucp2} and $b_0, c_0$ are as in Eq.~\eqref{eq:definition_b_c}.
 Then there is a constant 
 $
  C_2 = C_2(d,\delta_+, \lVert V_0 \rVert_\infty) 
 $
 such that for any closed interval $I \subset (- \infty, E_0]$ with $\lvert I \rvert \leq 2 \gamma_2$,
 any $\lambda > 0$,
 and any $L \in \NN_{\mathrm{odd}}$ with $L \geq 2 + \delta_+$, we have
 \[
  \EE \bigl ( \operatorname{Tr} P_{\omega, L}(I) \bigr)
  \leq
  C_2 \left( u_-^{-2} \gamma_2^{-4} (1 + E_0) \right)^{2^{1 + \frac{\log d}{\log 2}}}
  S_L(\lambda^{-1} \lvert I \rvert) \lvert \Lambda_L \rvert.
 \]
\end{theorem}
For $t \geq 0$ we define
\[
H_0(t) 
 := 
 H_0 
 +
 t \sum_{j \in \ZZ^d} u_j(\cdot - z_j),
\]
and set
\[
 E_0 (t) := \inf \sigma(H_0(t)) \quad \text{and} \quad 
 E_0(\infty)
 :=
 \lim_{t \to \infty} E_0(t)
 =
 \sup \left\{E_0(t) : t \geq 0 \right\} .
\]
\begin{theorem}
\label{thm:Klein3}
We have $E_0(\infty) > 0$. Let $E_1 \in ( 0, E_0(\infty))$ and set
\[
  \kappa_0
  =
  \sup_{s > 0 : E_0(s) \geq E_1}
  \frac{E_0(s) - E_1}{s} 
  >
  0 .
 \]
 Then 
 for any Borel set $B \subset (-\infty , E_1]$,
 any $\lambda > 0$, 
 any $L \in \NN_{\mathrm{odd}}$, 
 and almost all $\omega \in \Omega$ we have
 \begin{equation}\label{eq:Klein3_1}
  P_{\omega, L} (B)
  \Bigl(
   \sum_{j \in \Lambda_L \cap \ZZ^d} u_j(\cdot - z_j)
  \Bigr)
 P_{\omega, L} (B)
  \geq
  \kappa_0
  P_{\omega, L} (B) .
 \end{equation}
 Moreover, 
 for any closed interval $I \subset (- \infty, E_1]$,
 any $\lambda > 0$,
 and for any $L \in \NN_{\mathrm{odd}}$ with $L \geq 2 + \delta_+$, we have
 \begin{equation} \label{eq:Klein3_2}
  \EE \bigl( \operatorname{Tr} P_{\omega, L} (I) \bigr)
  \leq
  C_3
  \left( 
   \kappa_0^{-2} ( 1 + E_1)
  \right)^{2^{1 + \frac{\log d}{\log 2}}}
  S_L(\lambda^{-1} \lvert I \rvert) \lvert \Lambda_L \rvert,
 \end{equation}
 where $C_3 > 0$ is a constant depending on $d$, $\delta_+$, $\lVert V_0 \rVert_\infty$, $\lVert b_0 \rVert_\infty$, and $\lVert c_0 \rVert_\infty$.
 \end{theorem}
 The Wegner estimates from Theorems~\ref{thm:Klein1}, \ref{thm:Klein2} and \ref{thm:Klein3} can be used as an ingredient for the multiscale analysis \cite{FroehlichS-83,DreifusK-89,GerminetK-01,GerminetK-03,GerminetK-06}. Let us emphasize that the multiscale analysis requires that the concentration functions $S_L$, $L \in \NN$, are sufficiently regular, e.g.\ with a uniformly bounded density, or uniformly H\"older continuous, cf.\ the just mentioned references.
 The multiscale analysis is an induction argument to establish localization in its various manifestations (spectral, dynamical, etc.). While a Wegner estimate is required for the induction step, the so-called initial length scale estimate corresponds to the induction anchor. Hence, if the concentration functions $S_L$ are sufficiently regular, Theorems~\ref{thm:Klein1}, \ref{thm:Klein2} and \ref{thm:Klein3} will imply localization at energies where an appropriate initial scale estimate is satisfied.
 If the upper bound in the Wegner estimate becomes small at large disorder, the initial scale estimate will follow from the Wegner estimate at sufficiently large disorder as observed in \cite{DreifusK-89}, see also \cite{Kirsch-08}. 
\par
 The upper bounds in Theorems~\ref{thm:Klein1} and~\ref{thm:Klein2} grow as the disorder $\lambda$ increases.
 As discussed above, this is not sufficient to deduce an initial length-scale estimate and localization at large disorder.
 In contrast to that, the upper bound in Theorem~\ref{thm:Klein3} converges to $0$ as the disorder parameter $\lambda$ tends to $\infty$. Hence, an initial scale estimate and localization at large disorder follow, albeit only for energies below $E_0(\infty)$.
 Note that if a covering condition
 \begin{equation*}
  \sum_{j \in \ZZ^d} u_j(\cdot - z_j)
  \geq
  \epsilon 
  > 0
 \end{equation*}
 is satisfied, then $E_0(\infty) = \infty$ and the statement of Theorem~\ref{thm:Klein3} holds at all energies, see\ \cite{CombesH-94} in the case of vanishing magnetic field. In contrast, $E (\infty)$ might be finite if we do not assume a covering condition.
 Since Wegner estimates with a disorder dependence as in Theorem~\ref{thm:Klein3} provide a relatively simple path to localization at large disorder, it is natural to ask if such a disorder dependence can be expected at all energies, even if no covering condition is assumed.
 However, so far one was not able to prove such a Wegner estimate for alloy-type models with and without magnetic field above the threshold $E (\infty)$, cf.~\cite{Stollmann-10,BoutetDeMonvelLS-11,Klein-13}.
 
Our next theorem shows that this is indeed not possible. A disorder dependence as in Theorem~\ref{thm:Klein3} holds if and only if we consider  energy intervals below $E_0(\infty)$. 
 In particular this shows that at high energies and at high disorder there is a fundamental difference between alloy-type models with and without a covering condition. This is a new result, even in the special case of vanishing magnetic potential ($A_0 = 0$) and ergodic potential ($V_0$ periodic, $z_j = j$, $u_j = u_0$, and $\mu_j = \mu_0$).

\begin{theorem}
\label{thm:counterexample}
 Let $E_2 \in \RR$. The following are equivalent:
 \begin{enumerate}[(i)]
  \item 
  $E_2 \leq E_0(\infty)$.
  \item
  For all sufficiently large $L>0$, and all closed intervals $I \subset (- \infty, E_2]$, we have
  \begin{equation}
  \label{eq:equivalence}
   \EE \left( \operatorname{Tr} P_{\omega,L} (I) \right)
   \to 0 \quad\text{as}\quad\lambda \to\infty .
  \end{equation}
 \end{enumerate}
\end{theorem}

For the rest of this section we assume that the family $H_\omega$, $\omega \in \Omega$, is ergodic, i.e.\ $V_0$ and $A_0$ are periodic, and for all $j \in \ZZ^d$ we have $\mu_j = \mu_0$, $u_j = u_0$ and $z_j = j$. 
In this situation, Theorem~\ref{thm:counterexample} has an interpretation in terms of the integrated density of states (IDS). 
Since the main argument in this case is rather instructive, we present it here:
Let us denote by $N_\lambda : \RR \to [0,\infty)$ the IDS
of the family $H_\omega$, $\omega \in \Omega$. 
This is a distribution function satisfying
 \[
  N_\lambda(E)
  =
  \lim_{L \to \infty} \frac{\operatorname{Tr} P_{\omega, L}( (- \infty, E])}{\lvert \Lambda_L \rvert}.
 \]
 at all continuity points of $N_\lambda$ and for almost all $\omega \in \Omega$, cf. \cite{Pastur-71,Shubin-79,KirschM-82c}, see also \cite{Veselic-08} and the references therein. 
 Note that the Wegner estimates from Theorems~\ref{thm:Klein1},~\ref{thm:Klein2} and~\ref{thm:Klein3} imply local Lipschitz continuity of $N_{\lambda}$ at all $E \in \RR$ and for all $\lambda \in (0, \infty)$, if the measure $\mu_0$ is sufficiently regular.
 Moreover, if $\mathcal{U}^\infty := \RR^d \backslash \operatorname{supp} \sum_{j \in \ZZ^d} u(\cdot - j)$ is non-empty we will denote by $H_0^\infty$ the corresponding Dirichlet operator on $\mathcal{U}^\infty $, i.e. the unique self-adjoint extension of the operator $(- i \nabla + A_0)^2 + V_0$ on $C_0^\infty(\mathcal{U}^\infty) \subset L^2(\mathcal{U}^\infty)$. 
 For its IDS we use the notation $N_{0 , \infty}$.

 By Floquet theory, cf.~\cite{Zak-64, Sjoestrand-91}, we obtain $N_{0,\infty}(E) > 0$ for $E > E_0(\infty)$. 
 Furthermore,
  \begin{equation*}
  N_{0, \infty}(E) 
  \leq
  N_\lambda(E)
 \end{equation*}
 for all $\lambda > 0$ and all $E \in \RR$.
 This follows from
 \begin{equation}
 \label{eq:inequality_eigenvalues_IDS}
  \mu_k(H_{0,L}^\infty) 
  \geq 
  \mu_k(H_{\omega,L})
 \end{equation}
 where $\mu_k$ denotes the $k$-th eigenvalue of the corresponding operator, ordered increasingly and counting multiplicities and $H_{0,L}^\infty$ is the Dirichlet restriction of $H_0^\infty$ to $L^2(\mathcal{U}^\infty \cap \Lambda_L)$.
 For the convenience of the reader, we give a proof of Ineq.~\eqref{eq:inequality_eigenvalues_IDS} in the appendix.
Together with Theorem~\ref{thm:Klein3}, we found the following theorem.
 \begin{theorem} \label{thm:conterexample_ergodic}
  Let $H_\omega$, $\omega \in \Omega$, be ergodic.
  Then $\lim_{\lambda \to \infty} N_\lambda(E) \geq N_{0, \infty}(E)$ for all $E \in \RR$ and it holds that
  \begin{enumerate}[(i)]
   \item 
    $\lim_{\lambda \to \infty} N_{\lambda}(E) = 0$ if $E < E_0(\infty)$,
   \item
    $\lim_{\lambda \to \infty} N_\lambda(E) > 0$ if $E > E_0(\infty)$.
  \end{enumerate}
 \end{theorem}
 
 \begin{remark}
 \label{rem:phase_transition}
 Theorems~\ref{thm:counterexample} and~\ref{thm:conterexample_ergodic} show that the difference between the alloy-type model with and without a covering condition is not merely a technical issue. 
  Rather, we observe that the physical behaviour of the system at large disorder fundamentally differs between the phases $E < E_0 (\infty)$ and $E > E_0 (\infty)$.
 \end{remark}

%
%
%
%
%
%
%
%
\section{Quantitative unique continuation}
In~\cite{BorisovTV-15}, the authors prove quantitative unique continuation principles for second order elliptic partial differential expressions with variable coefficients.
Here, we formulate the special case where the leading term is the Laplacian.
Let
\begin{equation*} 
\Op u := - \Delta + b^\T \nabla u + c u 
\end{equation*}
with $b \in L^\infty(\RR^d; \CC^d)$ and $c \in L^\infty(\RR^d; \CC)$.
For $L > 0$ we denote by $\Dom (\Delta_L)$ the domain of the Laplace operator in $L^2 (\Lambda_L)$ subject to Dirichlet boundary conditions.
For $\Gamma \subset \RR^d$ open and $\psi \in L^2 (\Gamma)$ we denote by $\lVert \psi \rVert = \lVert \psi \rVert_\Gamma$ the usual $L^2$-norm of $\psi$. If $\Gamma' \subset \Gamma$ we use the notation $\lVert \psi \rVert_{\Gamma'} = \lVert \chi_{\Gamma'} \psi \rVert_\Gamma$.
The following theorem is a special case of Theorem 12 in \cite{BorisovTV-15}.

\begin{theorem}\label{thm:ucp_eigenfunction}
For all $L \in \NN_{\mathrm{odd}}$, all measurable and bounded $V : \Lambda_L \to \RR$, all  $\psi\in \Dom (\Delta_L)$ and $\zeta \in L^2 (\Lambda_L)$ satisfying $\lvert \Op \psi \rvert \leq \lvert V\psi \rvert + \lvert \zeta \rvert$ almost everywhere on $\Lambda_L$, all $\delta \in (0,1/2)$ and all sequences $X = (x_j)_{j \in \ZZ^d}$ such that $\ball{\delta}{x_j} \subset \Lambda_1(j)$ for all $j \in \ZZ^d$, we have
\begin{equation*} 
\lVert \psi \rVert_{S_{\delta,X} (L)}^2 + \delta^2 G^2 \lVert \zeta \rVert_{\Lambda_L}^2
\geq C_{\sfUC} \lVert \psi \rVert_{\Lambda_L}^2 ,
\end{equation*}
where 
 \begin{equation} \label{eq:definition-Sdelta}
 C_{\sfUC} = \delta^{N_1 \bigl( 1 +  \lVert V \rVert_\infty^{2/3} + \lVert b \rVert_\infty^{2} + \lVert c \rVert_\infty^{2/3} \bigr)}
 \quad\text{and}\quad
 S_{\delta , X} (L) = \bigcup_{j \in \ZZ^d} \ball{\delta}{x_j} \cap \Lambda_L .
\end{equation}
Here $N_1 \geq 1$ is a constant depending only on the dimension.
\end{theorem}
For $L > 0$ we define the differential operator $H_L : \Dom (\Delta_L) \to L^2 (\Lambda_L)$ by $H_L \psi = \Op \psi$. If 
\begin{equation}\label{eq:sa}
 b = \mathrm{i} \tilde b 
 \quad\text{and}\quad 
 c = \tilde c + \mathrm{i} \diver \tilde b / 2 
\end{equation}
for some bounded $\tilde b , \tilde c \in L^\infty (\RR^d)$, then $H_L$ is a self-adjoint operator in $L^2 (\Lambda_L)$. The following theorem is a special case of Theorems 13 and 14 in \cite{BorisovTV-15}.
\begin{theorem}\label{thm:qucp1}
Let \eqref{eq:sa} be satisfied. Then for all $L \in \NN_{\mathrm{odd}}$, all $E \in \RR$, all $\delta \in (0,1/2)$, all sequences $X = (x_j)_{j \in \ZZ^d}$ such that $\ball{\delta}{x_j} \subset \Lambda_1(j)$ for all $j \in \ZZ^d$, and all $\psi \in \operatorname{Ran} \chi_{[E-\gamma , E + \gamma]} (H_L)$ with
\[
 \gamma^2 = \delta^{N_2\bigl( 1 + \lvert E \rvert^{2/3} + \lVert b \rVert_\infty^{2} + \lVert c \rVert_\infty^{2/3} \bigr)}
\]
we have
\[
 \lVert \psi \rVert_{S_{\delta,X} (L)}^2  \geq \gamma^2 \lVert \psi \rVert_{\Lambda_L}^2 .
\]
Here $N_2\geq 1$ is a constant depending only on the dimension and $S_{\delta , X} (L)$ is as in Eq.~\eqref{eq:definition-Sdelta}.
\end{theorem}
As a corollary we obtain
\begin{theorem}[]\label{thm:qucp2}
Let \eqref{eq:sa} be satisfied, $E_0 \in \RR$, $\delta \in (0,1/2)$, and
\[
 \gamma^2 = \delta^{N_2\bigl( 1 + \lvert E_0 \rvert^{2/3} + \lVert b \rVert_\infty^{2} + \lVert c \rVert_\infty^{2/3} \bigr)} .
\]
Then for all $I \subset (-\infty , E_0]$ with $\lvert I \rvert \leq 2\gamma$, and all sequences $X = (x_j)_{j \in \ZZ^d}$ such that $\ball{\delta}{x_j} \subset \Lambda_1(j)$ for all $j \in \ZZ^d$, we have
\[
\chi_{I}(H_L) W_{\delta,X} (L) \chi_{I}(H_L)
\geq \gamma^2 \chi_{I}(H_L) .
\]
Here, $W_{\delta,X} (L)$ denotes the operator of multiplication with the characteristic function of the set $S_{\delta , X} (L)$ defined in Eq.~\eqref{eq:definition-Sdelta}.
\end{theorem}
%
%
%
%
%
%
%
%
\section{Proofs}\label{sec:proofs}
For $L \in \NN_{\mathrm{odd}}$, we define $U_L : \Lambda_L \to \RR$ and $W_L : \Lambda_L \to \RR$ by 
 \[
  U_L
  :=
  \sum_{j \in \ZZ^d \cap \Lambda_L} u_j(\cdot - z_j)
  \quad\text{and}\quad
   W_L
  := 
  \sum_{j \in \ZZ^d \cap \Lambda_L} \chi_{\ball{\delta_-}{z_j}} .
 \]
 Then $\lVert U_L \rVert_\infty \leq (2 + \delta_+)^d$
and $\lVert V_\omega \rVert_\infty \leq M \lVert U_L \rVert_\infty \leq M (2 + \delta_+)^d$ for almost all $\omega \in \Omega$.
Note that
 \[
  0 \leq W_L \leq u_- U_L, 
  \quad
  W_L^2 = W_L,
  \quad
  \text{and}
  \quad
  \lVert W_L \rVert_\infty = 1.
 \]

\subsection{Proof of Theorem~\ref{thm:Klein1}}
\begin{proof}[Proof of Theorem~\ref{thm:Klein1}]
 We follow the proof of Theorem~1.4 in~\cite{Klein-13} and assume $\lambda = 1$. 
 The general case $\lambda > 0$ follows by scaling the random variables $\omega_k \mapsto \lambda \omega_k$ which leads to $M \mapsto \lambda M$ and $S_L(\lvert I \rvert) \mapsto S_L(\lambda^{-1} \lvert I \rvert)$.
 Rewriting $H_\omega = - \Delta + b_0^T \nabla + c_0 + V_\omega$, where $b_0, c_0$ are defined in Eq.~\eqref{eq:definition_b_c} and given $E_0 > 0$, we define $\gamma_1$ as in \eqref{eq:gamma_1}.
 Then, by Theorem~\ref{thm:qucp2}, for all $L \in \NN_{\mathrm{odd}}$ and all intervals $I \subset (- \infty, E_0]$ with $\lvert I \rvert \leq 2 \gamma_1$ we have
 \[
  \chi_I(H_{0,L}) 
  \leq
  \gamma_1^{-2} \chi_I(H_{0,L}) W_L \chi_I(H_{0,L})
  \leq
  u_-^{-1} 
  \gamma_1^{-2} \chi_I(H_{0,L}) U_L \chi_I(H_{0,L}).
 \]
 Since $\sigma(H_0) \subset [0, \infty)$, we have $\sigma(H_{0,L}) \subset [0, \infty)$ and may assume $I \subset [0, E_0]$.
 Similarly to Theorem 1.4 in \cite{Klein-13}, we can now carefully follow the proof in~\cite{CombesHK-07}, keeping in mind that the proof therein is formulated for magnetic Schr\"odinger operators $H_0 = (- i \nabla A_0)^2 + V_0^2$.
 \par
 Note that the Combes-Thomas estimates (for magnetic Schr\"odinger operators) required in~\cite{CombesHK-07} depend only on $d$, $\delta_+$ and on $\lVert V_0 \rVert_\infty$, but not on the magnetic potential $A_0$, see Theorem~4.6 of~\cite{Shen-14}.
 Hence, the constant $C_1$ will only depend on $A_0$ via $\gamma_1$.
\end{proof}
\subsection{Proof of Theorem~\ref{thm:Klein2}}
We adapt the proof of \cite[Theorem 1.5]{Klein-13} to the magnetic setting.
\begin{lemma} \label{lem:Klein_3.1}
  Let $I \subset ( - \infty, E_0]$ be a closed interval and $L \in \NN_{\mathrm{odd}}$, $L \geq 2 + \delta_+$.
  Suppose that there is $\kappa > 0$ such that
  \[
  P_{\omega,L}(I) 
  U_L
  P_{\omega,L}(I)
  \geq
  \kappa
  P_{\omega,L}(I)
  \quad
  \text{with probability one}.
  \]
  Then there is a constant
  \[
  C_4 = C_4(d, \delta_+, \lVert V_0 \rVert_\infty)
  \]
  such that
  \[
  \EE \bigl( \operatorname{Tr} P_{\omega,L}(L) \bigr)
  \leq
  C_4 \left( \kappa^{-2} ( 1 + E_0) \right)^{2^{1 + \frac{\log d}{\log 2}}} S_L(\lambda^{-1} \lvert I \rvert) \lvert \Lambda_L \rvert.
  \]
\end{lemma}
\begin{proof}[Proof of Lemma~\ref{lem:Klein_3.1}]
 One can follow verbatim the proof of Lemma 3.1 in \cite{Klein-13} which partially relies on results from \cite{CombesHK-07}.
 The latter apply to the class of magnetic Schr\"odinger operators as considered in this note as well.
 The only issue to address is the dependence of $C_4$ on the various parameters.
 In~\cite{Klein-13}, Eqs.~(3.8) and (3.17), constants from Combes-Thomas estimates (for non-magnetic Schr\"odinger operators) which depend only on $d$, $\delta_+$ and on $\lVert V_0 \rVert_\infty$ enter the final constant $C_4$.
 Combes-Thomas estimates for magnetic Schr\"odinger operators do not depend on the magnetic potential $A_0$, see Theorem~4.6 of \cite{Shen-14}.
 Therefore the constant $C_4$ will not depend on the magnetic potential $A_0$.
\end{proof}
\begin{proof}[Proof of Theorem~\ref{thm:Klein2}]
 We follow the proof of Theorem 1.5 in~\cite{Klein-13}.
 Given $E_0 > 0$, define $\gamma_2$ as in Eq.~\eqref{eq:gamma_2}.
 Theorem~\ref{thm:qucp2} yields for all $\Lambda_L$ with $L \in \NN_{\mathrm{odd}}$, all intervals $I \subset (- \infty, E_0]$ with $\lvert I \rvert \leq 2 \gamma_2$ and almost all $\omega \in \Omega$ the estimate
\[
  \chi_I(H_{0,L}) 
  \leq
  \gamma_2^{-2} \chi_I(H_{0,L}) W_L \chi_I(H_{0,L})
  \leq
  u_-^{-1} 
  \gamma_2^{-2} \chi_I(H_{0,L}) U_L \chi_I(H_{0,L}).
 \]
 The statement of the theorem now follows from Lemma~\ref{lem:Klein_3.1}. 
\end{proof}

\subsection{Proof of Theorem~\ref{thm:Klein3}}
For the proof we shall need an abstract uncertainty relation for Schr\"odinger operators at the bottom of the spectrum which has been developed in \cite{BoutetDeMonvelLS-11}.
The following lemma is a slight generalization thereof, see Lemma 4.1 of \cite{Klein-13}.
\begin{lemma}\label{lemma:abstract_uncertainty}
  Let $H$ be a self-adjoint operator on a Hilbert space $\mathcal{H}$, bounded from below, and let $Y \geq 0$ be a bounded operator on $\mathcal{H}$. Let $H (t) = H + tY$ for $t \geq 0$, and set $E (t) = \inf \sigma (H (t))$ and $E (\infty) = \lim_{t \to \infty} E (t) = \sup_{t \geq 0} E (t)$. Suppose that $E (\infty) > E (0)$. For $E_1 \in (E (0) , E (\infty))$ let
  \[
  \kappa = \kappa (H , Y , E_1) = \sup_{s > 0 \colon E (s) > E _1} \frac{E (s) - E_1}{s} > 0 .
  \]
  Then for all bounded operators $V \geq 0$ on $\mathcal{H}$ and Borel sets $B \subset (-\infty , E_1]$ we have
  \[
  \chi_B (H + V) Y \chi_B (H + V) \geq \kappa \chi_B (H + V) .
  \]
\end{lemma}
We recall that $H_0 (t) = H_0 + t \sum_{j \in \ZZ^d} u_j(\cdot - z_j)$ for $t \geq 0$, $E_0 (t) = \inf \sigma (H_0 (t))$, and $E_0 (\infty) = \lim_{t \to \infty} E_0 (t) = \sup_{t \geq 0} E_0 (t)$.
\begin{lemma}\label{lemma:positive}
  For all $t \geq 0$ we have
  \[
  E_0 (t) \geq t u_- \delta_-^{N_1 \bigl( 1 + \bigl( \lVert V_0 \rVert_\infty + t u_- \bigr)^{2/3} + \lVert b_0 \rVert_\infty^{2} + \lVert c_0 \rVert_\infty^{2/3} \bigr)}.
  \]
  Hence,
  \[
  E_0 (\infty) \geq \sup_{t \in [0,\infty)} t \delta_-^{N_1 \bigl( 1 + \bigl( \lVert V_0 \rVert_\infty + t  \bigr)^{2/3} + \lVert b_0 \rVert_\infty^{2} + \lVert c_0 \rVert_\infty^{2/3} \bigr)} > 0 .
  \]
\end{lemma}

\begin{proof}
  We define $\tilde H_0(t) := H_0 + t u_- W$ for $t \geq 0$, $\tilde E_0 (t) = \inf \sigma (\tilde H_0 (t))$, and $\tilde E_0 (\infty) = \lim_{t \to \infty} \tilde E_0 (t) = \sup_{t \geq 0} \tilde E_0 (t)$.
  By assumption we have $E_0 (0) = \tilde E_0(0) = 0$.
  Moreover, $E_0 (\infty)$ and $\tilde E_0(\infty)$ are both well defined in $[0,\infty]$ by monotonicity. 
  Since $\tilde E_0(t) \leq E_0(t)$ for all $t \in [0, \infty]$, it suffices to show the statement of the lemma for $\tilde E_0(t)$.
  For $L \in \NN_{\mathrm{odd}}$ we denote by $\tilde H_{0,L} (t)$ the restriction of $\tilde H_0 (t)$ to $L^2 (\Lambda_L)$ subject to Dirichlet boundary conditions with domain $\Dom (\Delta_L)$, and set $\tilde E_{0,L} (t) = \inf \sigma (\tilde H_{0,L} (t))$. 
  Then $\tilde E_{0,L} (t) \geq \tilde E_0 (t) \geq 0$ for all $t \geq 0$. 
  Since $\lVert W_L \rVert_\infty = 1$ we have
  \[
  \tilde E_{0,L} (t) \leq d \left( \frac{\pi}{L} \right)^2 + \lVert V_0 \rVert_\infty + t u_- .
  \]
  Since $\tilde H_{0,L} (t)$ has purely discrete spectrum, there exists $\psi (t) \in \Dom (\Delta_L)$ with $\lVert \psi (t) \rVert = 1$ such that $\tilde H_{0,L} (t) \psi (t) = \tilde E_{0,L} (t) \psi (t)$. 
  We apply Theorem~\ref{thm:ucp_eigenfunction} and obtain for all $t \geq 0$
  \begin{align*}
  \bigl\langle \psi (t) , W_L \psi (t) \bigr\rangle
  &= \lVert \psi \rVert_{S_{\delta_- , Z} (L)}\\
  &\geq \delta_-^{N_1 \bigl( 1 + \bigl( d \pi^2 / L^2 + \lVert V_0 \rVert_\infty + t u_- \bigr)^{2/3} + \lVert b_0 \rVert_\infty^{2} + \lVert c_0 \rVert_\infty^{2/3} \bigr)} 
  \lVert \psi (t) \rVert_{\Lambda_L}^2 
  \end{align*}
  where
  \[
  S_{\delta_- , Z} (L) = \bigcup_{j \in \ZZ^d} \ball{\delta_-}{z_j} \cap \Lambda_L .
  \]
  The first statement of the lemma follows since $\lim_{L \to \infty} \tilde E_{0,L} (t) = \tilde E_0 (t)$. The second statement follows immediately.
\end{proof}
\begin{proof}[Proof of Theorem~\ref{thm:Klein3}]
  By Lemma~\ref{lemma:positive} we have $E_0 (\infty) > 0$ and hence $\kappa_0 > 0$.  For $L \in \NN_{\mathrm{odd}}$ we denote by $H_{0,L} (t)$ the restriction of $H_0 (t)$ to $L^2 (\Lambda_L)$ subject to Dirichlet boundary conditions with domain $\Dom (\Delta_L)$, and set $E_{0,L} (t) = \inf \sigma (H_{0,L} (t))$. Using
  $
   0 \leq E_0 (t) \leq E_{0,L} (t)
  $
  we obtain
  \[
  \kappa_{0,L} 
  := 
  \sup_{s > 0 : E_{0,L}(s) \geq E_1} \frac{E_{0,L}(s) - E_1}{s} 
  \geq \kappa_0 
  = 
  \sup_{s > 0 : E_0(s) \geq E_0(1)}
  \frac{E_0(s) - E_1}{s} > 0 .
  \]
  Hence, the assumptions of Lemma~\ref{lemma:abstract_uncertainty} are satisfied with $H = H_{0,L}$ and $Y = \sum_{j \in \ZZ^d}u_j (\cdot - z_j)$, and we obtain  Ineq.~\eqref{eq:Klein3_1}. Ineq.~\eqref{eq:Klein3_2} now follows from Ineq.~\eqref{eq:Klein3_1} and Lemma~\ref{lem:Klein_3.1}.
\end{proof}
\subsection{Proof of Theorem~\ref{thm:counterexample}}
 The implication (i) $\Rightarrow$ (ii) is the statement of Theorem~\ref{thm:Klein3}.
 In order to show the converse, we prove the contraposition:
 Let $E_2 > E_0(\infty)$, and $I = (- \infty, E_2]$.
 Note that for almost all $\omega \in \Omega$, we have $H_{\omega,L} \leq H_{\eta,L}$, where $\eta \in \Omega$, $\eta_k = M$ for all $k \in \ZZ^d$, hence
 \[
  \EE \left( \operatorname{Tr} P_{\omega,L} (I) \right)
  \geq
  \operatorname{Tr} P_{\eta,L} (I).
 \]
 Since 
 \[
  \lim_{t \to \infty} \lim_{L \to \infty} E_{0,L}(t) < E_2,
 \]
 there are $L_0 > 0$ and $\lambda_0 > 0$ such that for all $\lambda > \lambda_0$ we have
 \[
  1
  \leq
  \operatorname{Tr} P_{\eta,L_0}(I)
  \leq
  \EE \left( \operatorname{Tr} P_{\omega,L_0} (I) \right).
 \]
 Hence, \eqref{eq:equivalence} cannot hold.
%
%
%
%
%
%
%
\appendix
 \section{Proof of Ineq.~\eqref{eq:inequality_eigenvalues_IDS}}
 
 Recall that $H_{0,L}^\infty$ is the Dirichlet restriction of $H_{0,L}$ to $\mathcal{U}^\infty \cap \Lambda_L$, the complement of the support of the single-site potentials in $\Lambda_L$.

 \begin{lemma}
 \label{lem:forms}
  Let $L > 0$, $\lambda > 0$, and $\mathcal{U}^\infty \cap \Lambda_L \neq \emptyset$. 
  Then for all $\omega \in \Omega$
  \[
   \mu_k(H_{0,L}^\infty)
   \geq
   \mu_k(H_{\omega, L}).
  \]
 \end{lemma}
 
 \begin{proof}
  Let $H_{t, L}$ be the operator $H_{\omega, L}$ where all the random variables $\omega_j$ are set to $t$.
  Since $\mu_k(H_{t, L}) \geq \mu_k(H_{\omega, L})$ for all $t \geq \lambda M$, it suffices to show $\mu_k(H_{0,L}^\infty) \geq \mu_k(H_{t, L})$ for some $t \geq \lambda M$.
  We define a family of positive and closed quadratic forms $ \{ \mathcal{E}_t \}_{t > 0}$ on $W^{2,1}(\Lambda_L) \subset L^2(\Lambda_L)$ via
  \[
   \mathcal{E}_t(\phi) 
   = 
   \langle (- i \nabla - A_0) \phi, (- i \nabla - A_0) \phi \rangle 
   + 
   \langle \phi, ( V_0 + t \sum_j u_j(\cdot - j) )\phi \rangle.
  \]
  The form $\mathcal{E}_t$ is the unique quadratic form associated with the self-adjoint operator $H_{t, L}$ and the family $\{ \mathcal{E}_t \}_{t > 0}$ is monotonously increasing in the sense of quadratic forms.
  By a version of Kato's monotone convergence theorem for quadratic forms, see \cite[Theorem 4.1]{Simon-78}, there is a closed form $\mathcal{E}_\infty$ which is defined as
  \begin{align*}
   \mathcal{D} (\mathcal{E}_\infty) &= \left\{\phi \in \mathcal{D} (\mathcal{E}_1) \colon \sup_{t > 0} \mathcal{E}_t (\phi) < \infty \right\}, \\
   \mathcal{E}_\infty (\phi) &= \lim_{t \to \infty} \mathcal{E}_t (\phi)
  \end{align*}
  such that $\mathcal{E}_t \nearrow \mathcal{E}_\infty$ as $t \to \infty$ in strong resolvent sense.
  
  There is a caveat here concerning the notion of convergence in strong resolvent sense, cf.~\cite{Simon-78}:
  Since the form $\mathcal{E}_\infty$ is not densely defined on $L^2(\Lambda)$, one cannot define the corresponding resolvent in the usual manner.
  However, the form $\mathcal{E}_\infty$ can be restricted to the closed subspace $\overline{\mathcal{D}(\mathcal{E}_\infty)} \subset \Lambda_L$ on which it yields a densely defined, closed form $\mathcal{\tilde E}_\infty$.
  In fact, $\mathcal{\tilde E}_\infty$ is the unique form corresponding to $H_{0,L}^\infty$.
  Let now $T$ be the operator on $L^2(\Lambda_L)$ which is $(H_{0,L}^\infty + 1)^{-1}$ on $\overline{\mathcal{D}(\mathcal{E}_\infty)}$ and $0$ on its orthogonal complement.
  The nonzero eigenvalues of $T$ are the eigenvalues of $(H_{0,L}^\infty + 1)^{-1}$.
  Then the notion of strong resolvent convergence of $\mathcal{E}_t$ to $\mathcal{E}_\infty$ means in this situation that $(H_{t,L} + 1)^{-1} \phi \to T \phi$ for all $\phi \in L^2(\Lambda_L)$.
  
  This implies that the $k$-th eigenvalue (counted from above) of $(H_{t,L} + 1)^{-1}$ converges from above to the $k$-th eigenvalue (counted from above) of $(H_{0,L}^\infty + 1)^{-1}$. 
  \end{proof}
\section*{Acknowledgments}
 It is our privilege and honor to thank our mentor Ivan Veseli\'c for the idea to pursue this research direction.
 Furthermore, we gratefully acknowledge helpful discussions with Peter Stollmann.


\end{document}